\documentclass[letterpaper, 10 pt, conference]{ieeeconf}
\IEEEoverridecommandlockouts
\usepackage{graphics}
\usepackage{amsmath}
\usepackage{amssymb,amsfonts}
\usepackage{xcolor}
\usepackage{algorithmic}

\usepackage{nicefrac}
\usepackage{graphicx,subfig}
\usepackage{siunitx}
\usepackage{tikz,pgfplots}
\pgfplotsset{compat=newest}
\allowdisplaybreaks
\newtheorem{theorem}{Theorem}[section]
\newtheorem{remark}[theorem]{Remark}

\newtheorem{assumption}[theorem]{Assumption}

\newtheorem{definition}[theorem]{Definition}

\newtheorem{algorithm}[theorem]{Algorithm}

\newcommand{\uc}{u}

\newcommand\copyrighttext{%
  \footnotesize 
  \textcopyright 2024 IEEE. Personal use of this material is permitted.
  Permission from IEEE must be obtained for all other uses, in any current or future
  media, including reprinting/republishing this material for advertising or promotional
  purposes, creating new collective works, for resale or redistribution to servers or
  lists, or reuse of any copyrighted component of this work in other works.
}
\newcommand\copyrightnotice{%
    \begin{tikzpicture}[remember picture,overlay]
        \node[anchor=south,yshift=10pt] at (current page.south) {\fbox{\parbox{\dimexpr\textwidth-\fboxsep-\fboxrule\relax}{\copyrighttext}}};
    \end{tikzpicture}%
}

\title{\LARGE \bf
    Data-driven MPC with terminal conditions in the Koopman framework
}

\author{Karl Worthmann$^{1}$, Robin Strässer$^{2}$, Manuel Schaller$^{1}$, Julian Berberich$^{2}$, and Frank Allgöwer$^{2}$
\thanks{Karl Worthmann gratefully acknowledges funding by the German Research Foundation (DFG, project ID 507037103). 
F.\ Allgöwer is thankful that this work was funded by the Deutsche Forschungsgemeinschaft (DFG, German Research Foundation) under Germany's Excellence Strategy -- EXC 2075 -- 390740016 and within grant AL 316/15-1 -- 468094890.
R.\ Strässer thanks the Graduate Academy of the SC SimTech for its support.
}
\thanks{$^{1}$K.\ Worthmann and M.\ Schaller are with the Institute of Mathematics, Technische Universität Ilmenau, 99693 Ilmenau, Germany
        (e-mail: {\tt [karl.worthmann, manuel.schaller]@tu-ilmenau.de}).}%
\thanks{$^{2}$R.\ Strässer, J.\ Berberich, and F.\ Allgöwer are with the Institute for Systems Theory and Automatic Control, University of Stuttgart, 70550 Stuttgart, Germany
        (e-mail: {\tt [straesser, berberich, allgower]@ist.uni-stuttgart.de}).}%
}

\begin{document}
\maketitle
\copyrightnotice
\thispagestyle{empty}
\pagestyle{empty}

\begin{abstract}
    We investigate nonlinear model predictive control (MPC) with terminal conditions in the Koopman framework using extended dynamic mode decomposition (EDMD) to generate a data-based surrogate model for prediction and optimization. We rigorously show recursive feasibility and prove practical asymptotic stability w.r.t.\ the approximation accuracy. To this end, finite-data error bounds are employed. The construction of the terminal conditions is based on recently derived proportional error bounds to ensure the required Lyapunov decrease. Finally, we illustrate the effectiveness of the proposed data-driven predictive controller including the design procedure to construct the terminal region and controller.
\end{abstract}


\section{Introduction}
Model predictive control (MPC) is a well-established advanced control methodology. 
The underlying idea is to solve, at each time instant, a finite-horizon optimal control problem based on the most-recent state measurement to evaluate the controller, see, e.g.,~\cite{grune:pannek:2017}.
While MPC is attractive due to the simplicity of the underlying idea and its ability to handle constrained nonlinear multi-input, multi-output systems, some care is in order to ensure proper functioning, see, e.g., \cite{muller2017quadratic}.
To this end, often terminal conditions~\cite{chen1998quasi} are used to ensure recursive feasibility and asymptotic stability.

A key requirement to apply MPC is a reliable model to accurately predict the system behavior in dependence of the control.
In this regard, data-driven methods have recently gained popularity, see, e.g.,~\cite{martin:schon:allgower:2023b} and the references therein.
In this paper, we focus on methods based on extended dynamic mode decomposition (EDMD;~\cite{williams2015data}), whose theoretical underpinning is the Koopman framework~\cite{mezic:2005}, see also the recent review article~\cite{BrunBudi22}.
The Koopman operator replaces the original highly-nonlinear dynamics by linear dynamics in the infinite-dimensional space of observable functions. 
Then, a finite-dimensional data-driven surrogate model is generated by EDMD using linear regression.
This approach was extended to systems with inputs (EDMDc: EDMD with control; \cite{brunton:brunton:proctor:kutz:2016}). 
An alternative approach exploits control affinity to construct a bilinear surrogate model~\cite{surana:2016,williams:hemati:dawson:kevrekidis:rowley:2016}, which exhibits a superior performance if direct state-control couplings are present, see, e.g., \cite{peitz:otto:rowley:2020} and~\cite{folkestad2021koopman} for an application and a discussion within the scope of MPC.

Whereas convergence of EDMD in the infinite-data limit was shown in~\cite{korda:mezic:2018b}, an essential tool for a thorough controller design with guarantees are finite-data error bounds. Here, to the best of our knowledge, Igor Mezi\'{c} was the first to rigorously establish bounds on the estimation error~\cite{mezic2022numerical} for deterministic systems and ergodic sampling. 
Then, the authors in~\cite{zhang:zuazua:2023} provided error bounds based on i.i.d.\ sampling before the first finite-data bounds on the approximation error for control systems were derived in~\cite{nuske:peitz:philipp:schaller:worthmann:2023,schaller:worthmann:philipp:peitz:nuske:2023}. 
For deterministic and stochastic continuous- and discrete-time systems in Polish spaces with i.i.d.\ and ergodic sampling, error estimates under non-restrictive assumptions were derived in~\cite{philipp2024extended}.

EDMD has been successfully applied in MPC~\cite{peitz:otto:rowley:2020,korda:mezic:2018a}, see also~\cite{zhang:pan:scattolini:yu:xu:2022} for a tube-based approach. 
These approaches have been shown to perform well in applications, e.g., autonomous driving~\cite{kanai:yamakita:2022}. 
However, the first rigorous closed-loop analysis of MPC using EDMD in the prediction step was only recently provided in~\cite{bold:grune:schaller:worthmann:2023}.
The key step in controller design with closed-loop guarantees was to adapt the regression problem in EDMD and, then, deduce error bounds, which are proportional to the control and the (lifted) state~\cite{bold:grune:schaller:worthmann:2023,strasser2024safedmd}.

In this article, we present an EDMD-based MPC scheme with terminal conditions with closed-loop stability guarantees and verified domain of attraction.
In particular, we rigorously show recursive feasibility and practical asymptotic stability.
Here, the term \emph{practical} results from an accumulation of the error along the predicted trajectories within the optimization step of the MPC algorithm.
The design of the terminal conditions relies on our recently proposed controller design framework \emph{SafEDMD}~\cite{strasser2024safedmd}. Contrary to established approaches, e.g., based on the linear-quadratic regulator and EDMDc~\cite{brunton:brunton:proctor:kutz:2016}, we provide rigorous closed-loop guarantees based on tailored finite-data error bounds and numerically demonstrate a significantly improved closed-loop performance.

The outline is as follows: In Section~\ref{sec:problem:formulation}, we briefly recap SafEDMD in the Koopman framework.
In Section~\ref{sec:MPC:scheme}, we present our EDMD-based MPC scheme. Then, the respective analysis (recursive feasibility, practical asymptotic stability) and the constructive design of the terminal conditions is presented in Section~\ref{sec:MPC:analysis}. Finally, the results are numerically validated before conclusions are drawn in Section~\ref{sec:conclusions}.

\noindent \textbf{Notation}: For non-negative integers $a,b$, we use the notation $[a:b] := \{ i \in \mathbb{N}_0 \mid a \leq i \leq b\}$. 
For two sets $A, B \subset \mathbb{R}^n$, $A \oplus B = \{ \mathbf{z} \in \mathbb{R}^n \mid \exists\,\mathbf{x} \in A, \mathbf{y} \in B: \mathbf{z} = \mathbf{x} + \mathbf{y} \}$ is the Pontryagin sum.
A continuous function $\alpha: \mathbb{R}_{\geq 0} \rightarrow \mathbb{R}_{\geq 0}$ is called of class~$\mathcal{K}$ if it is strictly increasing and zero at zero.
If $\alpha \in \mathcal{K}$ is, in addition, unbounded, $\alpha$ is of class~$\mathcal{K}_\infty$.
A continuous function $\beta: \mathbb{R}_{\geq 0} \times \mathbb{N}_0 \to \mathbb{R}_{\geq 0}$ is said to be of class~$\mathcal{KL}$ if $\beta(\cdot,k) \in \mathcal{K}_\infty$ holds and $\beta(r,\cdot)$ is strictly monotonically decreasing with $\lim_{t \rightarrow \infty} \beta(r,t) = 0$.
The closed $\varepsilon$-ball around $\mathbf{x} \in \mathbb{R}^n$ is denoted by $\mathcal{B}_\varepsilon(\mathbf{x})$.

\section{The Koopman operator and SafEDMD}
\label{sec:problem:formulation}

We consider continuous-time dynamical control systems governed by the control-affine ordinary differential equation 
\begin{equation}\label{eq:system:dynamics}
    \dot{x}(t) = g_0(x(t)) + \sum_{i=1}^{m} g_i(x(t))\, \uc_i(t)
\end{equation}
with maps $g_i \in \mathcal{C}^1(\mathbb{R}^{n},\mathbb{R}^{n})$, $i \in [0:m]$. For the locally integrable control function $\uc \in L^\infty_{\operatorname{loc}}(\mathbb{R}_{\geq 0},\mathbb{R}^{m})$, we have (local) existence and uniqueness of the respective (Carath\'{e}odory) solution $x(\cdot;\hat{\mathbf{x}},\uc)$ emanating from $\hat{\mathbf{x}} \in \mathbb{R}^{n}$.
Typically, control functions~$\uc$ are implemented in a sampled-data fashion with zero-order hold, i.e., 
\begin{equation}\label{eq:control:zoh}
    \uc(t) \equiv \mathbf{u}_k \in \mathbb{R}^{m} \quad\text{ on }\quad [k \Delta t, (k+1) \Delta t), k \in \mathbb{N}_0,
\end{equation}
for sampling period $\Delta t > 0$. 
Invoking autonomy of the maps $g_i$, $i \in [0:m]$, we define the discrete-time system dynamics
\begin{equation}\label{eq:system:sampled-data}
    \mathbf{x}^+ \hspace*{-0.5mm} = \hspace*{-0.25mm} f(\mathbf{x},\mathbf{u}) \hspace*{-0.5mm} := \hspace*{-0.25mm} \mathbf{x} + \hspace*{-0.5mm} \int_{0}^{\Delta t} \hspace*{-2mm} g_0(x(t;\mathbf{x},u)) + G(x(t;\mathbf{x},u)) \mathbf{u}\,\mathrm{d}t
\end{equation}
with $G(x(t;\mathbf{x},u)) \mathbf{u} = \sum_{i=1}^{m} g_i(x(t;\mathbf{x},u)) u_i$ for the constant control function $u(t) \equiv \mathbf{u}$.
Hence,
\[
    \textbf{x}_{\mathfrak{u}}(k+1;\hat{\textbf{x}}) = f(\textbf{x}_{\mathfrak{u}}(k;\hat{\textbf{x}}),\mathbf{u}_k) = x((k+1) \Delta t; \hat{\mathbf{x}},u)
\]
holds with the control function~$u$ defined by~\eqref{eq:control:zoh} and the state $\textbf{x}_{\mathfrak{u}}(k+1;\hat{\textbf{x}})$ generated by the discrete-time system~\eqref{eq:system:sampled-data} using the control-input sequence $\mathfrak{u} = (\textbf{u}_i)_{i=0}^k$. In particular, the discrete time~$k \in \mathbb{N}_0$ corresponds to the continuous time~$k \Delta t$. 
Since we consider the stabilization task, we assume that $g_0$ vanishes at the origin, i.e., $g_0(\textbf{0}) = \textbf{0}$. Then, the origin is an equilibrium for $\textbf{u} = \textbf{0}$ in the dynamics~\eqref{eq:system:sampled-data}.

For given control function $u(t) \equiv \mathbf{u}$, the Koopman identity
\begin{equation}\label{eq:Koopman:identity}
    (\mathcal{K}^t_{\mathbf{u}} \psi)(\hat{\textbf{x}}) = \psi(x(t;\hat{\textbf{x}},u))
\end{equation}
holds for all \emph{observables} $\psi \in L^2(\mathbb{R}^{n},\mathbb{R})$, $t \geq 0$, $\hat{\mathbf{x}} \in \mathbb{R}^{n}$. 
Here, $\mathcal{K}^t_{\mathbf{u}}$ represents the \emph{Koopman operator} of the respective semigroup $(\mathcal{K}^t_{\mathbf{u}})_{t \geq 0}$ of bounded linear operators. 
It is noteworthy that the Koopman operator~$\mathcal{K}^t_{\mathbf{u}}$ maps an observable (function) $\psi$ to another observable $\mathcal{K}^t_{\mathbf{u}} \psi$.

To derive a data-driven surrogate model of the Koopman operator, we collect finitely many, linearly independent observables in the dictionary $\mathcal{D} := \{ \psi_k, k \in [0:M] \}$, whose span, $\mathbb{V} := \operatorname{span}( \mathcal{D} )$, forms an $(M+1)$-dimensional subspace. 
On the convex and compact set $\mathbb{X} \subset \mathbb{R}^{n}$ containing the origin in its interior, the compression $P_{\mathbb{V}} \mathcal{K}^t_{\mathbf{u}}|_{\mathbb{V}}$ is approximated by linear regression using $d \in \mathbb{N}$ samples $(\psi_k(\hat{\mathbf{x}}_j),\psi_k(x(t;\hat{\mathbf{x}}_j,u)))$, $j \in [1:d]$, see, e.g., \cite{BrunBudi22}. 
Further, we set $\psi_0 \equiv 1$, $\psi_k(\mathbf{x}) = x_k$ for $k \in [1:n]$, and $\psi_k(0) = 0$ with $\psi_k \in \mathcal{C}^2(\mathbb{R}^{n},\mathbb{R})$ for $k \in [n+1:M]$, resulting in
\begin{equation}\label{eq:lifting-function}
    \Phi(\mathbf{x}) = \begin{bmatrix} 
        1 & x_1 & \cdots & x_n & \psi_{n+1}(\mathbf{x}) & \cdots & \psi_M(\mathbf{x}) 
    \end{bmatrix}^{\hspace*{-1mm}\top}\hspace*{-2mm},
\end{equation}
$\mathbf{x} = \begin{bmatrix} \textbf{0}_n & I_{n} & 0_{n \times M-n} \end{bmatrix} \Phi(\mathbf{x})$, and some constant $L_\Phi > 0$ such that $\| \mathbf{x} \| \leq \| \Phi(\mathbf{x}) - \Phi(\mathbf{0}) \| \leq L_{\Phi} \| \mathbf{x} \|$ holds on~$\mathbb{X}$.

We require the following assumption originally proposed in~\cite{goswami:paley:2021}, which states that the compression of the Koopman operator coincides with the restriction $\mathcal{K}^t_{\textbf{u}}|_{\mathbb{V}}$, see also \cite{korda:mezic:2020} for sufficient conditions and~\cite{strasser:schaller:worthmann:berberich:allgower:2024a} for a detailed discussion.
\begin{assumption}[Invariance of~$\mathbb{V}$]\label{ass:invariant-dictionary}
    For any $\psi \in \mathbb{V}$ and $u(t) \equiv \mathbf{u} \in \mathbb{R}^{m}$, let $\psi(x(t;\cdot,u)) \in \mathbb{V}$ hold for all $t \geq 0$.
\end{assumption}

As motivated in~\cite{peitz:otto:rowley:2020} and rigorously shown in~\cite{philipp:schaller:worthmann:peitz:nuske:2023b}, the Koopman operator approximately inherits control affinity: 
\begin{equation}\nonumber
    \mathcal{K}^t_{\textbf{u}} \approx \mathcal{K}^t_{\textbf{0}} + \sum_{i=1}^{m} u_i (\mathcal{K}^t_{\textbf{e}_i} - \mathcal{K}^t_{\textbf{0}}),
\end{equation}
where $\textbf{e}_i$ stands for the $i$th unit vector, $i\in[1:m]$. Here, we apply SafEDMD as proposed in~\cite[Sec.~3]{strasser2024safedmd}, i.e., we learn a data-driven bilinear surrogate model of the form
\begin{equation}\label{eq:safedmd_surrogate}
    K_{\textbf{u}}^{\Delta t} = K_{\textbf{0}}^{\Delta t} + \sum_{i=1}^{m} u_i (K_{\textbf{e}_i}^{\Delta t} - K_{\textbf{0}}^{\Delta t})
\end{equation}
using $d$ i.i.d.\ data samples for $\textbf{u} = \textbf{0}$ and $\textbf{u} = \textbf{e}_i$, $i \in [1:m]$. 
Due to the constant observable $\psi_0(\textbf{x}) \equiv 1$, $\psi_k(\textbf{0}) = 0$ for $k \in [1:M]$, and $f(\textbf{0},\textbf{0}) = \textbf{0}$, we impose the following structure on the surrogate of the Koopman operator
\begin{equation}\nonumber
    K_{\textbf{0}}^{\Delta t} = \begin{bmatrix}
        1 & \textbf{0}^\top \\ \textbf{0} & A
    \end{bmatrix}, \quad K_{\textbf{e}_i}^{\Delta t} = \begin{bmatrix}
        1 & \textbf{0}^\top \\ \textbf{b}_{i} & B_i
    \end{bmatrix},\; i\in[1:m].
\end{equation}
The unknown matrices $A$, $B_i$ and the vector $\textbf{b}_{i}$ result from solving the linear regression problems 
\begin{align*}
    A & = \operatornamewithlimits{argmin}_{A\in\mathbb{R}^{M\times M}} \|Y_\mathbf{0} - A X_\mathbf{0}\|_\mathrm{F}, \\
    \begin{bmatrix} 
        \textbf{b}_{i} & B_i
    \end{bmatrix} & = \operatornamewithlimits{argmin}_{\textbf{b}_{i} \in \mathbb{R}^M, B_i\in\mathbb{R}^{M\times M}} \|Y_{\textbf{e}_i} - \begin{bmatrix}
        \textbf{b}_{i} & B_i
    \end{bmatrix} X_{\textbf{e}_i}\|_\mathrm{F}
\end{align*}
for $i \in [1:m]$, where $\|\cdot\|_\mathrm{F}$ is the Frobenius norm. The data matrices are 
$X_{\textbf{0}} = \begin{bmatrix}
    \textbf{0}_M & I_{M}
\end{bmatrix} \begin{bmatrix}
            \Phi(\textbf{x}_{\textbf{0},1}) & \cdots & \Phi(\textbf{x}_{\textbf{0},d})
\end{bmatrix}$,
\begin{align*}
        X_{\textbf{e}_i} &= \begin{bmatrix}
            \Phi(\textbf{x}_{\textbf{e}_i,1}) & \cdots & \Phi(\textbf{x}_{\textbf{e}_i,d})
        \end{bmatrix}, \ i \in [1:m], 
        \\
        Y_{\textbf{u}} &= \begin{bmatrix}
            \textbf{0}_M & \hspace*{-2mm}I_{M}
        \end{bmatrix} \begin{bmatrix}
            \Phi(x(\Delta t;\textbf{x}_{\textbf{u},1},u)) & \hspace*{-2mm}\cdots\hspace*{-2mm} & \Phi(x(\Delta t;\textbf{x}_{\textbf{u},d},u))
        \end{bmatrix}\hspace*{-1mm}.
\end{align*}
As shown in~\cite[Cor.~3.2]{strasser2024safedmd}, we have the following \emph{proportional} bound: for any probabilistic tolerance $\delta \in (0,1)$, amount of data $d_0 \in \mathbb{N}$, and sampling rate~$\Delta t$, there are constants $c_x,c_u \in \mathcal{O}(\nicefrac{1}{\sqrt{\delta d_0}}+\Delta t^2)$ such that
\begin{equation}\label{eq:error-bound-safEDMD}
    \|(\mathcal{K}_{\textbf{u}}^{\Delta t}\Phi)(\textbf{x}) - K_{\textbf{u}}^{\Delta t} \Phi(\textbf{x})\| \leq c_x \|\Phi(\textbf{x}) - \Phi(\textbf{0})\| + c_u \| \textbf{u} \|
\end{equation}
holds for all $\textbf{x} \in \mathbb{X}$ and $\textbf{u} \in \mathbb{U}$ with probability $1-\delta$ provided $d \geq d_0$, where $\mathbb{U}\subset\mathbb{R}^m$ is compact. 
In particular, the bound on the estimation error~\eqref{eq:error-bound-safEDMD} can be guaranteed for arbitrarily small constants $c_x,c_u>0$ for sufficiently many data points~$d_0$ and a small enough sampling rate~$\Delta t$.
For the construction of the terminal ingredients in the proposed MPC approach, it is crucial that~\eqref{eq:error-bound-safEDMD} formulates a proportional error bound, i.e., the right-hand side vanishes for $(\textbf{x},\textbf{u})=(\textbf{0},\textbf{0})$.

By compactness of~$\mathbb{X}$ and~$\mathbb{U}$, the bound~\eqref{eq:error-bound-safEDMD} on the estimation error implies that, for any $\varepsilon>0$ and probabilistic tolerance $\delta \in (0,1)$, there is an amount of data $d_0\in \mathbb{N}$ and a maximal time step $\Delta t_0> 0$ such that, for all $\Delta t \leq \Delta t_0$ and $d\geq d_0$, the SafEDMD surrogate model~\eqref{eq:safedmd_surrogate} satisfies the following bound with probability least $1-\delta$:
\begin{equation}\label{eq:error-bound}
        \| (\mathcal{K}_{\textbf{u}}^{\Delta t}\Phi)(\textbf{x}) - K_{\textbf{u}}^{\Delta t} \Phi(\textbf{x}) \| 
        \leq \varepsilon \qquad\,\forall\,\textbf{x} \in \mathbb{X}, \textbf{u} \in \mathbb{U}.
\end{equation}
For a sequence of control values $\mathfrak{u} = (\textbf{u}_{\kappa})_{{\kappa}=0}^{N-1} \subseteq \mathbb{U}$, the ${\kappa}$-step prediction, ${\kappa}\in[1:N]$, resulting from the surrogate model~\eqref{eq:safedmd_surrogate}, denoted by $x_{\mathfrak{u}}({\kappa};\hat{\textbf{x}})$, reads
\begin{equation}\label{eq:prediction:k_step}
    \begin{bmatrix} 
        \textbf{0}_n & I_{n} & 0_{n \times M-n} 
    \end{bmatrix} 
    K^{\Delta t}_{\textbf{u}_{{\kappa}-1}} \hspace*{-0.mm}\cdots K^{\Delta t}_{\textbf{u}_0} \Phi(\hat{\textbf{x}}).
\end{equation}
Then, using Inequality~\eqref{eq:error-bound} and the triangle inequality yields
\begin{equation}\label{eq:error-bound:k-step}
    \| (\mathcal{K}_{\textbf{u}_{{\kappa}-1}}^{\Delta t} \hspace*{-2.mm}\cdots \mathcal{K}_{\textbf{u}_0}^{\Delta t} \Phi)(\textbf{x}) - K^{\Delta t}_{\textbf{u}_{{\kappa}-1}} \hspace*{-2.mm}\cdots K^{\Delta t}_{\textbf{u}_0} \Phi(\textbf{x}) \| 
    \hspace*{-0.5mm}\leq\hspace*{-0.5mm} \varepsilon \hspace*{-0.5mm}\sum_{i=0}^{{\kappa}-1} L_\mathcal{K}^i
\end{equation}
for all $\textbf{x} \in \mathbb{X}$ and $\textbf{u} \in \mathbb{U}$ with $L_\mathcal{K} := \max_{\textbf{u} \in \mathbb{U}}  \| \mathcal{K}_\textbf{u}^{\Delta t} \| $.

In conclusion, SafEDMD allows us to derive a data-driven surrogate model capable of making multi-step predictions with arbitrary accuracy supposing that sufficient data is available and the sampling period~$\Delta t$ is small enough. 
The latter can be mitigated by collecting data for smaller~$\Delta t$ and, then, construct the predictors by applying the derived models multiple times. 
This explains why the sampling period~$\Delta t$ does not explicitly occur for generator-based surrogate models, see~\cite{bold:grune:schaller:worthmann:2023,strasser:schaller:worthmann:berberich:allgower:2024a}. 
However, from a practical viewpoint, operator-based models are desirable as they do not rely on derivative data, see~\cite{strasser2024safedmd} for a detailed discussion.

\section{EDMD-based MPC with terminal conditions}
\label{sec:MPC:scheme}

We propose a SafEDMD-based MPC controller with terminal conditions and present a notion of practical asymptotic stability, which will be verified in the subsequent section. A key feature is that the formulation of the terminal conditions is carried out in the lifted space, where we leverage the bilinear structure of the surrogate model for an explicit construction using SafEDMD in Subsection~\ref{sec:sub:SafEDMD}. 

In view of the constant observable function contained in $\Phi$, cf.~\eqref{eq:lifting-function}, we set $\widehat{\Phi}(\textbf{x}) := \begin{bmatrix} \textbf{0}_M & I_M \end{bmatrix} \Phi(\textbf{x})$. Then, $\widehat{\Phi}(\textbf{0}) = \textbf{0}$ holds.
Let the compact sets $\mathbb{X} \subseteq \mathbb{R}^{n}$ and $\mathbb{U} \subseteq \mathbb{R}^{m}$ represent the state and control constraints, respectively. We consider a terminal region $\mathbb{X}_{f} \subseteq \mathbb{X}$ with
\[
    \mathbb{X}_{f} := \{ \textbf{x} \in \mathbb{R}^{n} \mid \widehat{\Phi}(\textbf{x})^\top P^{-1} \widehat{\Phi}(\textbf{x}) \leq c \}
\]
parametrized by an $(M \times M)$-matrix $P = P^\top \succ 0$ and $c>0$. Further, we define the terminal cost $V_f:\mathbb{X}_f \to \mathbb{R}$ by
\[
    V_f(\textbf{x}):=\widehat{\Phi}(\textbf{x})^\top P^{-1} \widehat{\Phi}(\textbf{x}).
\]
We define admissibility of a control-input sequence. 
\begin{definition}[Admissibility]
    $\mathfrak{u} = (\textbf{u}_{\kappa})_{{\kappa}=0}^{N-1} \subset \mathbb{U}$ is said to be an admissible control sequence for $\hat{\textbf{x}} \in \mathbb{X} \subseteq \mathbb{R}^{n}$ and horizon~$N \in \mathbb{N}$, denoted by $\mathfrak{u} \in \mathcal{U}_N(\hat{\textbf{x}})$, if $x_{\mathfrak{u}}({\kappa};\hat{\textbf{x}}) \in \mathbb{X}$, ${\kappa} \in [1:N]$, and $x_{\mathfrak{u}}(N;\hat{\textbf{x}}) \in \mathbb{X}_{f}$ hold.
\end{definition}

Next, we suitably adapt the notion of admissibility to deal with approximation errors in the optimization step of the MPC algorithm, where predictions are conducted using the SafEDMD-based surrogate model. 
To this end, we use the following definition in dependence of the approximation accuracy~$\varepsilon$, $\varepsilon \in (0,\varepsilon_0]$, based on the ${\kappa}$-step error bound~\eqref{eq:error-bound:k-step}.
\begin{definition}[EDMD admissibility]\label{def:admissibility:EDMD}
    For $\varepsilon \in (0,\varepsilon_0]$, $\mathfrak{u} = (\textbf{u}_{\kappa})_{{\kappa}=0}^{N-1} \subset \mathbb{U}$ is said to be an EDMD-admissible control sequence for $\hat{\textbf{x}} \in \mathbb{X}$ and horizon~$N \in \mathbb{N}$ for the SafEDMD-based surrogate, denoted by $\mathfrak{u} \in \widehat{\mathcal{U}}_N(\hat{\textbf{x}})$, if the set inclusion 
    \begin{equation}\label{eq:admissibility:EDMD:trajectory}        
        [\textbf{0}_{n}\, I_{n}\, 0_{n\times M-n}] K^{\Delta t}_{\textbf{u}_{{\kappa}-1}} \hspace*{-1mm}\cdots  K^{\Delta t}_{\textbf{u}_0} \Phi(\hat{\textbf{x}}) \oplus \mathcal{B}_{{\bar{c}}({\kappa}) \varepsilon}(\textbf{0}) \subseteq \mathbb{X}
    \end{equation}
    holds for all ${\kappa} \in [1:N]$ with ${\bar{c}}({\kappa}) := \sum_{i=0}^{{\kappa}-1} L_\mathcal{K}^i$, where $\mathbb{X}$ is replaced by $\mathbb{X}_{f}$ for ${\kappa} = N$. 
\end{definition}

Based on this definition, we propose our EDMD-based MPC scheme with terminal conditions for the stage costs
\begin{equation}\label{eq:stagecost}
    \ell(\textbf{x},\textbf{u}) := \| \textbf{x} \|_Q^2 + \| \textbf{u} \|_R^2 := \textbf{x}^\top Q \textbf{x} + \textbf{u}^\top R \textbf{u}
\end{equation}
with $Q = Q^\top \succ 0$ and $R = R^\top \succ 0$ resembling ideas from robust MPC~\cite{mayne2005robust}.
\begin{algorithm}[Model predictive control with horizon~$N$]\label{alg:MPC}
    At each time~$k \in \mathbb{N}_0$:
    \begin{enumerate}
        \item [1)] Define $\hat{\textbf{x}} := \textbf{x}_{\mu_N^\varepsilon}(k)$ by the current state $\textbf{x}_{\mu_N^\varepsilon}(k) \in \mathbb{X}$.
        \item [2)] Let $\bar{\textbf{x}}_{\mathfrak{u}}({\kappa};\hat{\textbf{x}})$ be predicted using~\eqref{eq:prediction:k_step}. Minimize 
            \begin{equation}\nonumber
                 J_N(\hat{\textbf{x}},\mathfrak{u}) := \sum_{{\kappa}=0}^{N - 1} \ell(\bar{\textbf{x}}_{\mathfrak{u}}({\kappa};\hat{\textbf{x}}),\textbf{u}_{\kappa}) + V_f(\bar{\textbf{x}}_{\mathfrak{u}}(N;\hat{\textbf{x}}))
            \end{equation}
            w.r.t.\ $\mathfrak{u} = (\textbf{u}_\kappa)_{\kappa = 0}^{N-1} \in \widehat{\mathcal{U}}_N(\hat{\textbf{x}})$ 
            to compute the optimal sequence of control-input values $\mathfrak{u}^\star = (\textbf{u}^\star_{\kappa})_{{\kappa}=0}^{N-1}$ such that $V_N(\hat{\textbf{x}}) = J_N(\hat{\textbf{x}},\mathfrak{u}^\star)$ holds.
        \item [3)] Apply the feedback value $\mu^\varepsilon_N(\textbf{x}_{\mu_N^\varepsilon}(k)) := \textbf{u}^\star_0 \in \mathbb{U}$.
    \end{enumerate}
\end{algorithm}

\noindent Our goal is to show recursive feasibility and practical asymptotic stability of EDMD-based MPC as formulated in Algorithm~\ref{alg:MPC}. 
To this end, we recall \cite[Def.~11.9]{grune:pannek:2017} --~practical asymptotic stability w.r.t.\ the approximation error~$\varepsilon$, i.e., that the behavior of the closed-loop (dynamical) system resembles an asymptotically-stable one until an arbitrarily small neighborhood of the origin is reached, where the size of the neighborhood depends on the approximation accuracy~$\varepsilon$.
\begin{definition}[Practical asymptotic stability]\label{def:stability:practical}
    The origin is said to be practically asymptotically stable (PAS) w.r.t.\ the approximation error on the set $A \subseteq \mathbb{X}$ containing the origin in its interior if there exists $\beta \in \mathcal{K}\mathcal{L}$ such that: for each $r>0$, there is $ \varepsilon_0 > 0$ such that, for all $\varepsilon \in (0,\varepsilon_0]$ satisfying condition~\eqref{eq:error-bound}, 
    the solution $x_{\mu_N^\varepsilon}(\cdot,\hat{\textbf{x}})$ of 
    \begin{align}\label{eq:ex_cl}
        \textbf{x}_{\mu_N^\varepsilon} ({k}+1) = f(\textbf{x}_{\mu_N^\varepsilon}({k}),\mu_N^\varepsilon(\textbf{x}_{\mu_N^\varepsilon}({k}))),
    \end{align}
    $\textbf{x}_{\mu_N^\varepsilon}(0) = \hat{\textbf{x}} \in A$, with $f$ from~\eqref{eq:system:sampled-data} fulfills $\textbf{x}_{\mu_N^\varepsilon}({k};\hat{\textbf{x}}) \in A$ and 
    \begin{align}\label{eq:def_stab_practical_2}
        \| \textbf{x}_{\mu_N^\varepsilon}({k};\hat{\textbf{x}}) \| \leq \max \{\beta(\| \hat{\textbf{x}} \|,{k}),r\} \qquad\forall\,{k} \in \mathbb{N}_0
    \end{align}
    for the feedback law~$\mu_N^\varepsilon$ defined in Step~3) of Algorithm~\ref{alg:MPC}. 
\end{definition}

\section{Analysis of EDMD-based MPC and construction of the terminal conditions}
\label{sec:MPC:analysis}

We first show recursive feasibility and practical asymptotic stability of the origin w.r.t.\ the MPC closed loop resulting from Algorithm~\ref{alg:MPC}. 
The proof is based on the following assumption, whose (constructive) verification is discussed in Subsection~\ref{sec:sub:SafEDMD}. 
To this end, we construct a suitable terminal region and terminal cost using SafEDMD~\cite{strasser2024safedmd} in combination with the proportional error bound~\eqref{eq:error-bound-safEDMD}.
\begin{assumption}[Terminal conditions]\label{ass:terminal:condition}
    Let a continuous sampled-data controller~$\mu: \mathbb{X}_{f} \rightarrow \mathbb{U}$ with $\mu(\textbf{0}) = \textbf{0}$ be given such that $\mathbb{X}_{f}$ is rendered invariant w.r.t.\ the discrete-time dynamics~\eqref{eq:system:sampled-data} and the Lyapunov decrease
    \begin{equation}\label{eq:SafEDMD-Lyapunov-decay}
        V_f(f(\textbf{x},\mu(\textbf{x}))) \leq V_f(\textbf{x}) - \ell(\textbf{x},\mu(\textbf{x}))
    \end{equation}
    holds for all $\textbf{x} \in \mathbb{X}_f$ with stage cost $\ell(\textbf{x},\textbf{u}) = \| \textbf{x} \|_Q^2 + \| \textbf{u} \|_R^2$.
\end{assumption}

\subsection{MPC closed-loop analysis}
\label{sec:sub:PAS}

In this part, we provide our main theoretical result. We show that, assuming initial feasibility, the MPC closed loop is well defined (recursively feasible) and exhibits practical asymptotic stability of the origin w.r.t.\ the MPC closed loop.
\begin{theorem}
\label{thm:MPC-guarantees}
    Let Assumptions~\ref{ass:invariant-dictionary} and \ref{ass:terminal:condition} hold.
    Then, the MPC closed loop is recursively feasible and the origin is practically asymptotically stable w.r.t.\ the approximation error on the set $\mathbb{X}$ in the sense of Definition~\ref{def:stability:practical}.
\end{theorem}
\begin{proof}
    Let $r > 0$ be given. W.l.o.g., we assume $\mathcal{B}_r = \mathcal{B}_r(\textbf{0}) \subseteq \mathbb{X}_f$; otherwise, we reduce~$r$ until this inclusion holds while the norm bound in Definition~\ref{def:stability:practical} follows also for the original $r$. 
    In view of \eqref{eq:error-bound:k-step}, we choose a sufficiently small, but fixed estimation error~$\varepsilon>0$ such that the following construction can be conducted.
    
    Let $\eta > 0$ be such that $f(\textbf{x},\textbf{u}) \in \mathcal{B}_r$ holds for all $\textbf{x} \in S := V_f^{-1}[0,\eta] \oplus {\bar{c}}(N) \varepsilon \subseteq \mathcal{B}_r$ and all control values $\textbf{u} \in \mu(S) := \{ \textbf{v} \in \mathbb{R}^{m} \mid \exists\,\textbf{x} \in S: \mu(\textbf{x}) = \textbf{v} \}$, where $V_f^{-1}[0,\eta]$ is the sub-level set $\{ \textbf{x} \mid V_f(\textbf{x}) \leq \eta \}$.
    Furthermore, for $\hat{\textbf{x}} \in A$, let the robust EDMD admissible control-input sequence $\mathfrak{u}^\star = (\textbf{u}_{\kappa}^\star)_{{\kappa}=0}^{N-1} \in \widehat{\mathcal{U}}_N(\hat{\textbf{x}})$ be optimal in Step~2) of Algorithm~\ref{alg:MPC} meaning, in particular, that $\hat{\textbf{x}}$ is initially feasible. 
    Then, we show that the (shifted and prolonged) control sequence 
    \begin{equation}\label{eq:input-shifted:candidate}
        \mathfrak{u}^+ := \begin{pmatrix} \textbf{u}_1^\star & \cdots & \textbf{u}_{N-1}^\star & \mu(\textbf{x}_{\mathfrak{u}^\star}(N;\hat{\textbf{x}})) \end{pmatrix}
    \end{equation}
    is feasible for the successor state $\textbf{x}^+ = f(\hat{\textbf{x}},\mu^\varepsilon_N(\hat{\textbf{x}}))$ of the discrete-time dynamics~\eqref{eq:system:sampled-data} and yields a Lyapunov decrease w.r.t.\ the respective (optimal) value function~$V_N$ if $\hat{\textbf{x}} \notin S$. Otherwise $\textbf{x}^+$ is contained in $\mathcal{B}_r$ per construction.
    
    Feasibility of $\textbf{x}^+$ and $\textbf{x}_{\mathfrak{u}^+}({\kappa};\textbf{x}^+)$, ${\kappa} \in [1:N-1]$, directly follows from Assumption~\ref{ass:invariant-dictionary} and the constraint tightening in Definition~\ref{def:admissibility:EDMD}. 
    Since, in the compact set $\mathbb{X}_f \setminus V_f^{-1}[0,\eta]$,
    \[
        \min_{ \textbf{x} \in \mathbb{X}_f \setminus V_f^{-1}[0,\eta] } V_f(\textbf{x}) - V_f(f(\textbf{x},\mu(\textbf{x}))) > 0
    \]
    holds for the \emph{minimal decrease}, continuity of $\mu$ implies
    \[
        \xi := \min_{\textbf{x} \in \mathbb{X}_f \setminus V_f^{-1}[0,\eta]}\ \max_{ \textbf{y} \in Y }\ V_f(\textbf{x}) - V_f(f(\textbf{x},\mu(\textbf{y}))) > 0,
    \]
    $Y = (\mathbb{X}_f \cap \mathcal{B}_{{\bar{c}}(N) \varepsilon}(\textbf{x})) \setminus V_f^{-1}[0,\eta]$, for sufficiently small~$\varepsilon$.
    In conclusion, we get a decrease by applying $\textbf{u}^+_{N-1}$ in the last prediction step, whenever the predicted penultimate state is not contained in the $\eta$-ball centered at the origin, which completes the proof of recursive feasibility in view of the prelude before defining~$\mathfrak{u}^\star$.

In the following, we follow the line of reasoning in~\cite{limon2009input} and verify a suitable decay of the optimal value function~$V_N$, which serves as a Lyapunov function. 
To this end, for the ${k}$th state~$\hat{\textbf{x}} = \textbf{x}^\text{MPC}_{\mu^{\varepsilon}_N}({k})$ of the MPC closed-loop trajectory, let $\mathfrak{u}^\star = \mathfrak{u}^\star(\hat{\textbf{x}}) = (\textbf{u}^\star_{\kappa})_{{\kappa}=0}^{N-1}$ again denote the optimal control-input sequence.
Then, for the successor state $\textbf{x}^+ = f(\hat{\textbf{x}},\textbf{u}^\star(0))$, we define the candidate control-input sequence at time ${k}+1$ by~\eqref{eq:input-shifted:candidate}.
This yields 
\begin{align*}
    &V_N(\textbf{x}^+) - V_N(\hat{\textbf{x}}) + \ell(\textbf{x}_{\mathfrak{u}^\star}(0),\textbf{u}^\star_0) \\
    &\leq \sum_{{\kappa}=0}^{N-2} \left( \ell(\textbf{x}_{\mathfrak{u}^+}({\kappa}),\textbf{u}^+_{\kappa}) - \ell(\textbf{x}_{\mathfrak{u}^\star}({\kappa}+1),\textbf{u}^\star_{{\kappa}+1} \right) \\
    &\quad + \underbrace{\ell(\textbf{x}_{\mathfrak{u}^+}(N-1),\textbf{u}^+_{N-1}) + V_f(\textbf{x}_{\mathfrak{u}^+}(N))}_{\leq V_f(\textbf{x}_{\mathfrak{u}^+}(N-1)) + \| P^{-1} \| \tilde{c} L_\Phi \varepsilon} - V_f(\textbf{x}_{\mathfrak{u}^\star}(N))
\end{align*}
with $\textbf{x}_{\mathfrak{u}^+}({\kappa}) = \textbf{x}_{\mathfrak{u}^+}({\kappa};\textbf{x}^+)$ and $\textbf{x}_{\mathfrak{u}^\star}({\kappa}) = \textbf{x}_{\mathfrak{u}^\star}({\kappa};\hat{\textbf{x}})$ for ${\kappa} \in [0:N]$ in view of the assumed Lyapunov decrease~\eqref{eq:SafEDMD-Lyapunov-decay}, 
where the term $\| P^{-1} \| \tilde{c} L_\Phi \varepsilon$ results from the difference of $V_f(\textbf{x}_{\mathfrak{u}^+}(N)) - V_f(f(\textbf{x}_{\mathfrak{u}^+}(N-1),\textbf{u}_{N-1}^+))$ analogously to the following calculations:
For vectors $\textbf{a},\textbf{b} \in \mathbb{R}^{n}$ and a matrix $M \in \mathbb{R}^{n \times n}$, the estimate 
\begin{align*}
    \| \textbf{a} \|_M^2 \hspace*{-0.25mm}-\hspace*{-0.25mm} \| \textbf{b} \|_M^2 \hspace*{-0.25mm}=\hspace*{-0.25mm} (\textbf{a}+\textbf{b})^{\hspace*{-0.5mm}\top} \hspace*{-0.5mm}M (\textbf{a}-\textbf{b}) \hspace*{-0.25mm}\leq\hspace*{-0.25mm} \|M\| \|\textbf{a}+\textbf{b}\| \|\textbf{a} \hspace*{-0.25mm}-\hspace*{-0.25mm} \textbf{b}\|
\end{align*}
holds. Hence, every difference in the stage cost in the sum (${\kappa} \in [0:N-2]$) can be estimated by
\begin{align*}
    \|Q\| \| \textbf{x}_{\mathfrak{u}^+}({\kappa}) + \textbf{x}_{\mathfrak{u}^\star}({\kappa}+1) \| \| \textbf{x}_{\mathfrak{u}^+}({\kappa}) - \textbf{x}_{\mathfrak{u}^\star}({\kappa}+1) \|,
\end{align*}
where the second factor can be uniformly estimated on the compact set~$\mathbb{X}$ by ${\tilde{c}}$, while the third summand is uniformly bounded by ${\bar{c}}({\kappa}+1) \varepsilon$. 

Next, we estimate the term $V_f(\textbf{x}_{\mathfrak{u}^+}(N-1)) - V_f(\textbf{x}_{\mathfrak{u}^\star}(N))$ by using the definition $V_f(\textbf{x})=\widehat{\Phi}(\textbf{x})^\top P^{-1} \widehat{\Phi}(\textbf{x})$ and the Lipschitz continuity of $\widehat{\Phi}$ by 
\begin{align*}
    & {\tilde{c}} \|P^{-1}\| L_\Phi \| \textbf{x}_{\mathfrak{u}^+}(N-1) - \textbf{x}_{\mathfrak{u}^\star}(N) \|
    \leq {\tilde{c}} \|P^{-1}\| L_\Phi {\bar{c}}(N) \varepsilon.
\end{align*}
Combining both estimates and using $\ell(\textbf{x}_{\mathfrak{u}^\star}(0),\textbf{u}^\star_0) \geq \| \hat{\textbf{x}} \|_Q^2$, we get
\begin{align*}
    & V_N(\textbf{x}^+) - V_N(\hat{\textbf{x}})  \\
    \leq & {\tilde{c}} \varepsilon \Big[ \| P^{-1} \| L_\Phi {(1+\bar{c}}(N)) + \| Q \| {\sum_{\kappa=0}^{N-2}}
    {\bar{c}}({\kappa}+1) \Big] - \| \hat{\textbf{x}} \|_Q^2, 
\end{align*}
i.e., the desired Lyapunov decrease outside the sub-level set~$V_f^{-1}[0,\eta]$ contained in~$\mathcal{B}_r$. 
Using standard arguments~\cite{grune:pannek:2017}, this completes the proof of practical asymptotic stability.
\end{proof}

We emphasize that the used constraint tightening ensuring recursive feasibility relies on a simple propagation of a one-step error bound to bound the deviation between nominal and true state along the prediction horizon~\cite{marruedo2002input}. Whereas the present article serves as a starting point, more advanced robust MPC techniques may be applied to reduce conservatism, see, e.g.,~\cite{kohler2020computationally} or~\cite{rakovic2021robust} and the references therein. 

Naturally, the question arises whether the proposed MPC controller may indeed render the closed loop asymptotically stable. 
Based on the above line of reasoning, this cannot be easily shown.
To see this, one may \emph{trace back} (estimate from above similarly to the derived Inequality~\eqref{eq:error-bound:k-step}) the term $\| \textbf{x}_{\mathfrak{u}^+}({\kappa}) - \textbf{x}_{\mathfrak{u}^\star}({\kappa}+1)\|$, ${\kappa} \in [1:N-1]$, to $\| \textbf{x}_{\mathfrak{u}^+}(0) - \textbf{x}_{\mathfrak{u}^\star}(1;\hat{\textbf{x}}) \|$ and, then, invoke the proportional error bound~\eqref{eq:error-bound-safEDMD}. However, this yields a linear term in norm, which cannot be bounded by the quadratic \emph{decrease} established in our proof arbitrarily close to the origin.
In conclusion, we cannot directly compensate the linear error bound locally around the origin, preventing to conclude \emph{asymptotic} stability.

\subsection{Construction of the terminal conditions using SafEDMD}
\label{sec:sub:SafEDMD}

Next, we show how to design the terminal conditions of the proposed MPC Algorithm~\ref{alg:MPC} such that Assumption~\ref{ass:terminal:condition} holds. 
In particular, we exploit the proportional error bound~\eqref{eq:error-bound-safEDMD} of SafEDMD to design a locally stabilizing control law for all $\textbf{x} \in \mathbb{X}_{f}$.
To this end, we leverage~\cite[Thm.~4.1]{strasser2024safedmd} to infer the sampled-data controller
\begin{equation}\label{eq:controller}
    \mu(\textbf{x}) = (I-L_w(\Lambda^{-1} \otimes \widehat{\Phi}(\textbf{x})))^{-1} L P^{-1}_\mu \widehat{\Phi}(\textbf{x}),
\end{equation}
where $\otimes$ stands for the Kronecker product and the matrices $L\in\mathbb{R}^{m\times M}$, $L_w\in\mathbb{R}^{m\times Mm}$, $0\prec\Lambda=\Lambda^\top\in\mathbb{R}^{m\times m}$, and $0\prec P_\mu=P_\mu^\top\in\mathbb{R}^{M\times M}$ are chosen such that two linear matrix inequalities are satisfied, see~\cite[Ineq.~(17),~(18)]{strasser2024safedmd}.
More precisely, $\mu$ renders the terminal region $\mathbb{X}_f$ invariant w.r.t. the discrete-time system~\eqref{eq:system:sampled-data} and guarantees $\|\widehat{\Phi}(f(\textbf{x},\mu(\textbf{x})))\|^2_{P^{-1}_\mu} < \|\widehat{\Phi}(\textbf{x})\|^2_{P^{-1}_\mu}$ for any $\textbf{x} \in \mathbb{X}_f \setminus \{ \textbf{0} \}$ if~\eqref{eq:error-bound-safEDMD} holds. 
To satisfy~\eqref{eq:SafEDMD-Lyapunov-decay}, we exploit the prior strict inequality, i.e., there exists $\epsilon_x,\epsilon_u>0$ such that
\begin{align*}
    \|\widehat{\Phi}(f(\textbf{x},\mu(\textbf{x})))\|^2_{P^{-1}_\mu}
    & \leq \|\widehat{\Phi}(\textbf{x})\|^2_{P^{-1}_\mu} - \epsilon_x\| \textbf{x} \|^2-\epsilon_u\|\mu(\textbf{x})\|^2 \\
    &\leq \|\widehat{\Phi}(\textbf{x})\|^2_{P^{-1}_\mu} - \epsilon \ell(\textbf{x},\mu(\textbf{x}))
\end{align*}
with $\epsilon=\min\left\{\nicefrac{\epsilon_x}{\|Q\|},\nicefrac{\epsilon_u}{\|R\|}\right\}$ for all $\textbf{x} \in \mathbb{X}_{f}$.
Hence, defining $P = \epsilon P_\mu$, $c=\epsilon^{-1}$ with the constructed terminal conditions using SafEDMD ensures the required Assumption~\ref{ass:terminal:condition}. We conclude this section highlighting an advantage of dual-mode MPC combining Algorithm~\ref{alg:MPC} and SafEDMD.

\begin{remark}[Dual-mode MPC]\label{rk:dual-mode-MPC}
    Besides providing a suitable data-driven way to construct the terminal ingredients, SafEDMD can be also used to obtain exponential stability~\cite[Thm.~4.1]{strasser2024safedmd}.
    Once the practically-stable region is reached (which is guaranteed by Theorem~\ref{thm:MPC-guarantees}), one may switch to the stabilizing terminal controller. The underlying reason is that, in the region of practical stability, the uncertainty attached to the SafEDMD predictions outweighs the advantages of its prediction capability.
\end{remark}

\section{Simulation}
\label{sec:numerics}
In the following, we evaluate our proposed MPC scheme. We consider an undamped inverted pendulum 
\begin{align*}
    \dot{x}_1 = x_2, \quad \dot{x}_2 = \frac{g}{l}\sin(x_1) - \frac{b}{ml^2}x_2 + \frac{1}{ml^2}u
\end{align*}
and apply SafEDMD to obtain a data-driven bilinear surrogate model. 
Algorithm~\ref{alg:MPC} is implemented in Matlab using MPCTools~\cite{risbeck:rawlings:2016} with its interface to the nonlinear optimization software CasADi~\cite{andersson:gillis:horn:rawlings:diehl:2019}.
For the simulation, we choose $b=0.5$, $l=1$, $m=1$, and $g=9.81$, and collect $d=6000$ data samples for each constant control input $u(t)\equiv \bar{u}$ with $\bar{u}\in\{0,1\}$, where we sample uniformly from $\mathbb{X}:=[-15,15]^2$ with the sampling rate $\Delta t = 0.01$. 
We impose the control constraint $\mathbb{U}:=[-25,25]$.
Further, we choose $Q=I$ and $R=0.1I$ in the stage cost~\eqref{eq:stagecost}. 
We set $\Phi(x) = \begin{bmatrix}
    1 & x_1 & x_2 & \sin(x_1)
\end{bmatrix}^\top$ and design the terminal ingredients based on~\cite[Thm.~4.1]{strasser2024safedmd} with $c_x=c_u=\SI{3e-4}{}$ for the learning error bound~\eqref{eq:error-bound-safEDMD} with $\varepsilon$ in~\eqref{eq:error-bound}, $S_z=0$, $R_z=2.5$, and $Q_z$ according to ~\cite[Procedure~8]{strasser:schaller:worthmann:berberich:allgower:2024a}.
We deploy the MPC scheme with a horizon $N=1.1/\Delta t=110$ to stabilize the unstable equilibrium at the origin. 
Fig.~\ref{fig:exmp-pendulum-sin} depicts the closed-loop behavior under the proposed controller. 
    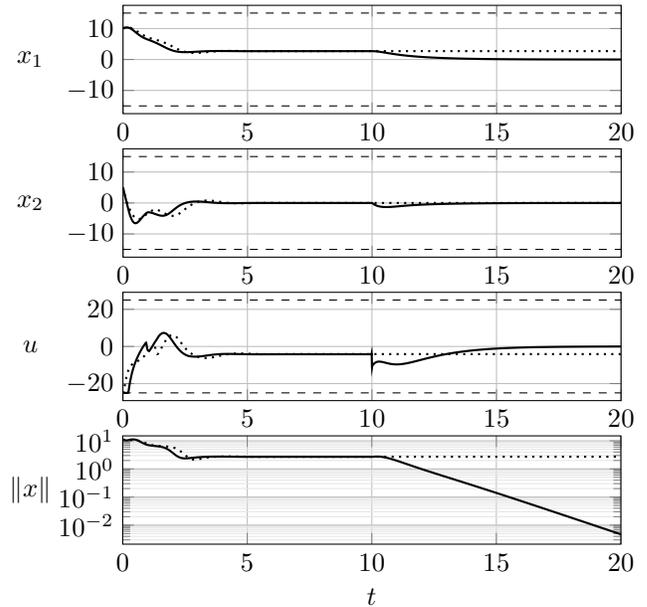
\begin{figure}[t]
        \centering
        \begin{tikzpicture}[%
            /pgfplots/every axis x label/.style={at={(0.5,0)},yshift=-20pt},%
            /pgfplots/every axis y label/.style={at={(0,0.5)},xshift=-35pt,rotate=0},%
          ]%
            \begin{axis}[
                name={x1},%
                legend pos= south east,
                xmin=0,
                xmax=20,
                ylabel=$x_1$,
                ymin=-17.5,
                ymax=17.5,
                grid=both,
                width = 0.95\columnwidth,
                height = 0.35\columnwidth,
                minor grid style={gray!20},
                unbounded coords = jump,
                restrict y to domain =-100:100,
            ]
                \addplot[black,thick,smooth] table [x index=0,y index=1] {data/exmp-inverse-pendulum-sin.dat};
                \addplot[black,thick,dotted] table [x index=0,y index=7] {data/exmp-inverse-pendulum-sin.dat};
                \draw[dashed] (axis cs:0,15) -- (axis cs:20,15);
                \draw[dashed] (axis cs:0,-15) -- (axis cs:20,-15);
            \end{axis}
            \begin{axis}[
                name={x2},%
                at={(x1.below south west)},anchor={north west},%
                legend pos= south east,
                xmin=0,
                xmax=20,
                ylabel=$x_2$,
                ymin=-17.5,
                ymax=17.5,
                grid=both,
                width = 0.95\columnwidth,
                height = 0.35\columnwidth,
                minor grid style={gray!20},
                unbounded coords = jump,
                restrict y to domain =-100:100,
            ]
                \addplot[black,thick,smooth] table [x index=0,y index=2] {data/exmp-inverse-pendulum-sin.dat};
                \addplot[black,thick,dotted] table [x index=0,y index=8] {data/exmp-inverse-pendulum-sin.dat};
                \draw[dashed] (axis cs:0,15) -- (axis cs:20,15);
                \draw[dashed] (axis cs:0,-15) -- (axis cs:20,-15);
            \end{axis}
            \begin{axis}[
                name={u},%
                at={(x2.below south west)},anchor={north west},%
                legend pos= south east,
                xmin=0,
                xmax=20,
                ylabel=$u$,
                ymin=-29.16667,
                ymax=29.16667,
                grid=both,
                width = 0.95\columnwidth,
                height = 0.35\columnwidth,
                minor grid style={gray!20},
                unbounded coords = jump,
                restrict y to domain =-100:100,
            ]
                \addplot[black,thick,smooth] table [x index=0,y index=3] {data/exmp-inverse-pendulum-sin.dat};
                \addplot[black,thick,dotted] table [x index=0,y index=9] {data/exmp-inverse-pendulum-sin.dat};
                \draw[dashed] (axis cs:0,25) -- (axis cs:20,25);
                \draw[dashed] (axis cs:0,-25) -- (axis cs:20,-25);
            \end{axis}
            \begin{semilogyaxis}[
                name={eerror},%
                at={(u.below south west)},anchor={north west},%
                legend pos= south east,
                xlabel=$t$,
                xmin=0,
                xmax=20,
                ylabel=$\|x\|$,
                ymax=15,
                yminorgrids = true,
                grid=both,
                width = 0.95\columnwidth,
                height = 0.35\columnwidth,
                minor grid style={gray!20},
                unbounded coords = jump,
            ]
                \addplot[black,thick,smooth] table [x index=0,y index=10] {data/exmp-inverse-pendulum-sin.dat};
                \addplot[black,thick,dotted] table [x index=0,y index=12] {data/exmp-inverse-pendulum-sin.dat};
            \end{semilogyaxis}
        \end{tikzpicture}
        \vspace*{-\baselineskip}
        \setbox1=\hbox{\begin{tikzpicture}[baseline]
            \draw[black,thick] (0,.6ex)--++(0.95em,0);
        \end{tikzpicture}}
        \setbox2=\hbox{\begin{tikzpicture}[baseline]
            \draw[black,thick,dotted] (0,.6ex)--++(0.95em,0);
        \end{tikzpicture}}
        \setbox3=\hbox{\begin{tikzpicture}[baseline]
            \draw[black,thick,dashed] (0,.6ex)--++(0.95em,0);
        \end{tikzpicture}}
        \vspace*{-0.5\baselineskip}
        \caption{Closed-loop results of the proposed MPC controller (\usebox1) with input constraints (\usebox3), where we switch to the terminal control law after $t=10$. MPC results for a linear Koopman surrogate 
        model (\usebox2) are included for comparison.}
        \label{fig:exmp-pendulum-sin}
        \vspace*{-1.5\baselineskip}
    \end{figure}

As expected due to Theorem~\ref{thm:MPC-guarantees}, the state is \emph{practically} stabilized, i.e., converges to a set-point close to the origin within the designed terminal region.
Notably, by following Remark~\ref{rk:dual-mode-MPC}, i.e., switching to the stabilizing terminal control law after reaching the (invariant) terminal region at $t\approx 10$, we can remove the offset and obtain a dual-mode MPC which \emph{asymptotically} stabilizes the origin.
For comparison, we apply MPC based on a \emph{linear} Koopman model (L-MPC) based on EDMDc~\cite{brunton:brunton:proctor:kutz:2016}, which is a commonly used Koopman-based control technique~\cite{korda:mezic:2018a}. 
Here, L-MPC stabilizes the nonlinear system with a remaining offset to the origin, but offers no guarantees w.r.t.\ closed-loop stability of the nonlinear system.
Although we only have shown that the proposed MPC scheme yields asymptotic stability when operated in dual-mode with SafEDMD, we briefly illustrate that the MPC control law may also be asymptotically stabilizing without switching. 
Applying SafEDMD for the observables $\Phi_1(x)= \begin{bmatrix} 1& x_1& x_2& \sin(x_1)&x_2& \cos(x_1)\end{bmatrix}^\top$, the MPC scheme asymptotically stabilizes the origin (see Fig.~\ref{fig:exmp-pendulum-sin-x2cos}).
    Here, L-MPC still results in an offset and oscillating closed-loop behavior when applied to the nonlinear system.
    \begin{figure}[tb]
        \centering
        \begin{tikzpicture}[%
            /pgfplots/every axis x label/.style={at={(0.5,0)},yshift=-20pt},%
            /pgfplots/every axis y label/.style={at={(0,0.5)},xshift=-35pt,rotate=0},%
          ]%
            \begin{axis}[
                name={u},%
                legend pos= south east,
                xmin=0,
                xmax=10,
                ylabel=$u$,
                ymin=-29.16667,
                ymax=29.16667,
                grid=both,
                width = 0.95\columnwidth,
                height = 0.35\columnwidth,
                minor grid style={gray!20},
                unbounded coords = jump,
                restrict y to domain =-100:100,
            ]
                \addplot[black,thick,smooth] table [x index=0,y index=3] {data/exmp-inverse-pendulum-sin-x2cos.dat};
                \addplot[black,thick,dotted] table [x index=0,y index=9] {data/exmp-inverse-pendulum-sin-x2cos.dat};
                \draw[dashed] (axis cs:0,25) -- (axis cs:20,25);
                \draw[dashed] (axis cs:0,-25) -- (axis cs:20,-25);
            \end{axis}
            \begin{semilogyaxis}[
                name={eerror},%
                at={(u.below south west)},anchor={north west},%
                legend pos= south east,
                xlabel=$t$,
                xmin=0,
                xmax=10,
                ylabel=$\|x\|$,
                ymax=20,
                yminorgrids = true,
                grid=both,
                width = 0.95\columnwidth,
                height = 0.35\columnwidth,
                minor grid style={gray!20},
                unbounded coords = jump,
            ]
                \addplot[black,thick,smooth] table [x index=0,y index=10] {data/exmp-inverse-pendulum-sin-x2cos.dat};
                \addplot[black,thick,dotted] table [x index=0,y index=12] {data/exmp-inverse-pendulum-sin-x2cos.dat};
            \end{semilogyaxis}
        \end{tikzpicture}
        \vspace*{-1.5\baselineskip}
        \setbox1=\hbox{\begin{tikzpicture}[baseline]
            \draw[black,thick] (0,.6ex)--++(0.95em,0);
        \end{tikzpicture}}
        \setbox2=\hbox{\begin{tikzpicture}[baseline]
            \draw[black,thick,dotted] (0,.6ex)--++(0.95em,0);
        \end{tikzpicture}}
        \setbox3=\hbox{\begin{tikzpicture}[baseline]
            \draw[black,thick,dashed] (0,.6ex)--++(0.95em,0);
        \end{tikzpicture}}
        \caption{Closed-loop results using $\Phi_1$ and the proposed MPC controller (\usebox1) compared to L-MPC (\usebox2).
        }
        \label{fig:exmp-pendulum-sin-x2cos}
        \vspace*{-1\baselineskip}
    \end{figure}
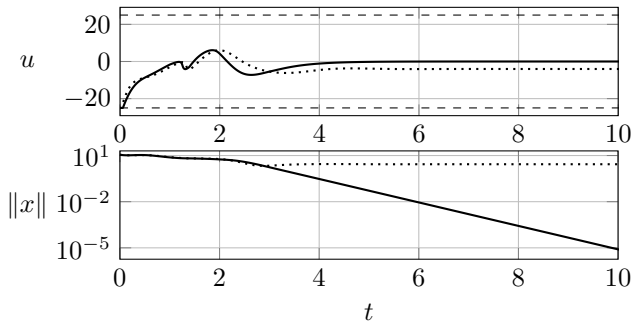

\section{Conclusions and outlook}
\label{sec:conclusions}
We proposed a data-driven MPC scheme with terminal conditions, where a variant of EDMD is used to generate a bilinear surrogate of the nonlinear system. 
The terminal region and costs are constructed using the recently proposed SafEDMD learning architecture. 
We rigorously showed practical asymptotic stability w.r.t.\ the MPC closed loop. 
Further, employing a dual-mode MPC approach based on SafEDMD yields exponential stability. 
The results are illustrated by a numerical example showing the efficacy in comparison with MPC based on linear models obtained from EDMDc.

Future work will be devoted to the removal of Assumption~\ref{ass:invariant-dictionary} by using the uniform bounds on the approximation error recently proposed for kernel EDMD in~\cite{kohne2024infty}.

\addtolength{\textheight}{-12cm}

\bibliographystyle{IEEEtran}
\bibliography{literature.bib}

\begin{thebibliography}{10}
\providecommand{\url}[1]{#1}
\csname url@samestyle\endcsname
\providecommand{\newblock}{\relax}
\providecommand{\bibinfo}[2]{#2}
\providecommand{\BIBentrySTDinterwordspacing}{\spaceskip=0pt\relax}
\providecommand{\BIBentryALTinterwordstretchfactor}{4}
\providecommand{\BIBentryALTinterwordspacing}{\spaceskip=\fontdimen2\font plus
\BIBentryALTinterwordstretchfactor\fontdimen3\font minus
  \fontdimen4\font\relax}
\providecommand{\BIBforeignlanguage}[2]{{%
\expandafter\ifx\csname l@#1\endcsname\relax
\typeout{** WARNING: IEEEtran.bst: No hyphenation pattern has been}%
\typeout{** loaded for the language `#1'. Using the pattern for}%
\typeout{** the default language instead.}%
\else
\language=\csname l@#1\endcsname
\fi
#2}}
\providecommand{\BIBdecl}{\relax}
\BIBdecl

\bibitem{grune:pannek:2017}
L.~Gr{\"u}ne and J.~Pannek, \emph{Nonlinear model predictive control}.\hskip
  1em plus 0.5em minus 0.4em\relax Springer, 2017.

\bibitem{muller2017quadratic}
M.~A. M{\"u}ller and K.~Worthmann, ``Quadratic costs do not always work in
  {MPC},'' \emph{Automatica}, vol.~82, pp. 269--277, 2017.

\bibitem{chen1998quasi}
H.~Chen and F.~Allg{\"o}wer, ``A quasi-infinite horizon nonlinear model
  predictive control scheme with guaranteed stability,'' \emph{Automatica},
  vol.~34, no.~10, pp. 1205--1217, 1998.

\bibitem{martin:schon:allgower:2023b}
T.~Martin, T.~B. Sch{\"o}n, and F.~Allg{\"o}wer, ``Guarantees for data-driven
  control of nonlinear systems using semidefinite programming: A survey,''
  \emph{Annual Reviews in Control}, p. 100911, 2023.

\bibitem{williams2015data}
M.~O. Williams, I.~G. Kevrekidis, and C.~W. Rowley, ``{A Data–Driven
  Approximation of the Koopman Operator: Extending Dynamic Mode
  Decomposition},'' \emph{J. Nonlinear Science}, vol.~25, pp. 1307--1346, 2015.

\bibitem{mezic:2005}
I.~Mezi{\'c}, ``Spectral properties of dynamical systems, model reduction and
  decompositions,'' \emph{Nonlinear Dyn.}, vol.~41, pp. 309--325, 2005.

\bibitem{BrunBudi22}
S.~L. Brunton, M.~Budišić, E.~Kaiser, and J.~N. Kutz, ``{Modern Koopman
  Theory for Dynamical Systems},'' \emph{SIAM Review}, vol.~64, no.~2, p.
  229–340, 2022.

\bibitem{brunton:brunton:proctor:kutz:2016}
S.~L. Brunton, B.~W. Brunton, J.~L. Proctor, and J.~N. Kutz, ``{Koopman}
  invariant subspaces and finite linear representations of nonlinear dynamical
  systems for control,'' \emph{PloS one}, vol. 11(2), pp. 1--19, 2016.

\bibitem{surana:2016}
A.~Surana, ``{Koopman} operator based observer synthesis for control-affine
  nonlinear systems,'' in \emph{Proc. 55th IEEE Conference on Decision and
  Control (CDC)}, 2016, pp. 6492--6499.

\bibitem{williams:hemati:dawson:kevrekidis:rowley:2016}
M.~O. Williams, M.~S. Hemati, S.~T. Dawson, I.~G. Kevrekidis, and C.~W. Rowley,
  ``Extending data-driven {Koopman} analysis to actuated systems,''
  \emph{IFAC-PapersOnLine}, vol.~49, no.~18, pp. 704--709, 2016.

\bibitem{peitz:otto:rowley:2020}
S.~Peitz, S.~E. Otto, and C.~W. Rowley, ``Data-driven model predictive control
  using interpolated {Koopman} generators,'' \emph{SIAM Journal on Applied
  Dynamical Systems}, vol.~19, no.~3, pp. 2162--2193, 2020.

\bibitem{folkestad2021koopman}
C.~Folkestad and J.~W. Burdick, ``Koopman {NMPC}: {K}oopman-based learning and
  nonlinear model predictive control of control-affine systems,'' in
  \emph{Proc. IEEE International Conference on Robotics and Automation (ICRA)},
  2021, pp. 7350--7356.

\bibitem{korda:mezic:2018b}
M.~Korda and I.~Mezi{\'c}, ``On convergence of extended dynamic mode
  decomposition to the {Koopman} operator,'' \emph{J. Nonlinear Science},
  vol.~28, no.~2, pp. 687--710, 2018.

\bibitem{mezic2022numerical}
I.~Mezi{\'c}, ``On numerical approximations of the {K}oopman operator,''
  \emph{Mathematics}, vol.~10, no.~7, p. 1180, 2022.

\bibitem{zhang:zuazua:2023}
C.~Zhang and E.~Zuazua, ``A quantitative analysis of {Koopman} operator methods
  for system identification and predictions,'' \emph{Comptes Rendus.
  M{\'e}canique}, vol. 351, no.~S1, pp. 1--31, 2023.

\bibitem{nuske:peitz:philipp:schaller:worthmann:2023}
F.~N{\"u}ske, S.~Peitz, F.~Philipp, M.~Schaller, and K.~Worthmann,
  ``Finite-data error bounds for {Koopman}-based prediction and control,''
  \emph{J. Nonlinear Science}, vol. 33:14, 2023.

\bibitem{schaller:worthmann:philipp:peitz:nuske:2023}
M.~Schaller, K.~Worthmann, F.~Philipp, S.~Peitz, and F.~N{\"u}ske, ``Towards
  reliable data-based optimal and predictive control using extended {DMD},''
  \emph{IFAC-PapersOnLine}, vol. 56(1), pp. 169--174, 2023.

\bibitem{philipp2024extended}
F.~Philipp, M.~Schaller, S.~Boshoff, S.~Peitz, F.~N{\"u}ske, and K.~Worthmann,
  ``Variance representations and convergence rates for data-driven
  approximations of {K}oopman operators,'' \emph{arXiv:2402.02494}, 2024.

\bibitem{korda:mezic:2018a}
M.~Korda and I.~Mezi{\'c}, ``Linear predictors for nonlinear dynamical systems:
  {Koopman} operator meets model predictive control,'' \emph{Automatica},
  vol.~93, pp. 149--160, 2018.

\bibitem{zhang:pan:scattolini:yu:xu:2022}
X.~Zhang, W.~Pan, R.~Scattolini, S.~Yu, and X.~Xu, ``Robust tube-based model
  predictive control with {K}oopman operators,'' \emph{Automatica}, vol.
  137:110114, 2022.

\bibitem{kanai:yamakita:2022}
M.~Kanai and M.~Yamakita, ``Linear model predictive control with lifted
  bilinear models by {K}oopman-based approach,'' \emph{SICE Journal of Control,
  Measurement, and System Integration}, vol.~15, no.~2, pp. 162--171, 2022.

\bibitem{bold:grune:schaller:worthmann:2023}
L.~Bold, L.~Gr{\"u}ne, M.~Schaller, and K.~Worthmann, ``Data-driven {MPC} with
  stability guarantees using extended dynamic mode decomposition,'' \emph{IEEE
  Transactions on Automatic Control}, 2025, available online:
  \url{https://doi.org/10.1109/TAC.2024.3431169}.

\bibitem{strasser2024safedmd}
R.~Str{\"a}sser, M.~Schaller, K.~Worthmann, J.~Berberich, and F.~Allg{\"o}wer,
  ``{SafEDMD}: A certified learning architecture tailored to data-driven
  control of nonlinear dynamical systems,'' \emph{arXiv:2402.03145}, 2024.

\bibitem{goswami:paley:2021}
D.~Goswami and D.~A. Paley, ``Bilinearization, reachability, and optimal
  control of control-affine nonlinear systems: A {Koopman} spectral approach,''
  \emph{IEEE Transactions on Automatic Control}, vol.~67, no.~6, pp.
  2715--2728, 2021.

\bibitem{korda:mezic:2020}
M.~Korda and I.~Mezi{\'c}, ``Optimal construction of {Koopman} eigenfunctions
  for prediction and control,'' \emph{IEEE Transactions on Automatic Control},
  vol.~65, no.~12, pp. 5114--5129, 2020.

\bibitem{strasser:schaller:worthmann:berberich:allgower:2024a}
R.~Str{\"a}sser, M.~Schaller, K.~Worthmann, J.~Berberich, and F.~Allgöwer,
  ``{Koopman}-based feedback design with stability guarantees,'' \emph{IEEE
  Transactions on Automatic Control}, 2024, available online
  \url{https://doi.org/10.1109/TAC.2024.3425770}.

\bibitem{philipp:schaller:worthmann:peitz:nuske:2023b}
F.~Philipp, M.~Schaller, K.~Worthmann, S.~Peitz, and F.~N{\"u}ske, ``Error
  analysis of kernel {EDMD} for prediction and control in the {Koopman}
  framework,'' \emph{arXiv:2312.10460}, 2023.

\bibitem{mayne2005robust}
D.~Q. Mayne, M.~M. Seron, and S.~Rakovi{\'c}, ``Robust model predictive control
  of constrained linear systems with bounded disturbances,'' \emph{Automatica},
  vol.~41, no.~2, pp. 219--224, 2005.

\bibitem{limon2009input}
D.~Limon, T.~Alamo, D.~M. Raimondo, D.~M. De~La~Pe{\~n}a, J.~M. Bravo,
  A.~Ferramosca, and E.~F. Camacho, ``Input-to-state stability: a unifying
  framework for robust model predictive control,'' \emph{Nonlinear Model
  Predictive Control: Towards New Challenging Applications}, pp. 1--26, 2009.

\bibitem{marruedo2002input}
D.~L. Marruedo, T.~Alamo, and E.~F. Camacho, ``Input-to-state stable {MPC} for
  constrained discrete-time nonlinear systems with bounded additive
  uncertainties,'' in \emph{Proc. 41st IEEE Conference on Decision and Control
  (CDC)}, vol.~4, 2002, pp. 4619--4624.

\bibitem{kohler2020computationally}
J.~K{\"o}hler, R.~Soloperto, M.~A. M{\"u}ller, and F.~Allg{\"o}wer, ``A
  computationally efficient robust model predictive control framework for
  uncertain nonlinear systems,'' \emph{IEEE Transactions on Automatic Control},
  vol.~66, no.~2, pp. 794--801, 2020.

\bibitem{rakovic2021robust}
S.~V. Rakovi{\'c}, ``Robust model predictive control,'' in \emph{Encyclopedia
  of systems and control}.\hskip 1em plus 0.5em minus 0.4em\relax Springer,
  2021, pp. 1965--1975.

\bibitem{risbeck:rawlings:2016}
\BIBentryALTinterwordspacing
M.~J. Risbeck and J.~B. Rawlings, ``{MPCTools}: Nonlinear model predictive
  control tools for {CasADi},'' 2016. [Online]. Available:
  \url{https://bitbucket.org/rawlings-group/octave-mpctools}
\BIBentrySTDinterwordspacing

\bibitem{andersson:gillis:horn:rawlings:diehl:2019}
J.~A.~E. Andersson, J.~Gillis, G.~Horn, J.~B. Rawlings, and M.~Diehl,
  ``{CasADi} -- {A} software framework for nonlinear optimization and optimal
  control,'' \emph{Math. Prog. Comp.}, vol.~11, no.~1, pp. 1--36, 2019.

\bibitem{kohne2024infty}
F.~K{\"o}hne, F.~Philipp, M.~Schaller, A.~Schiela, and K.~Worthmann,
  ``{$L^{\infty}$}-error bounds for approximations of the {K}oopman operator by
  kernel extended dynamic mode decomposition,'' \emph{arXiv:2403.18809}, 2024.

\end{thebibliography}

\end{document}